\documentclass[preprint,11pt]{elsarticle}

\usepackage{amssymb,amsmath}
\newtheorem{definition}{Definition}
\newtheorem{lemma}{Lemma}
\newtheorem{proposition}{Proposition}
\newtheorem{theorem}{Theorem}

\newtheorem{remark}{Remark}
\newtheorem{algorithm_logic}{Algorithm}
\newproof{proof}{Proof}

\journal{Communications in Nonlinear Science and Numerical Simulation Supports}

\usepackage{booktabs}
\usepackage[hidelinks]{hyperref}
\usepackage{tikz}
\usetikzlibrary{shapes.geometric, arrows.meta, positioning, calc, shadows}

\begin{document}

\begin{frontmatter}

\title{$N$-Player Binary Games with Unidirectional Dependencies: Cycle
  Stability and Induced Indifference}

%Authors, affiliations address.
\author{
  Sánchez-Sáez, José-María,
  Odishelidze, Nana and
  Criado-Aldeanueva, Francisco\corref{cor1}
}
\cortext[cor1]{fcaldeanueva@ctima.uma.es}
\affiliation{organization={Department of Didactics of Mathematics, Faculty of
    Education, University of Malaga},%Department and Organization
            addressline={Bulevar Louis Pasteur, 25}, 
            city={Malaga},
            postcode={29010}, 
            state={Andalucía},
            country={Spain}}

\affiliation{organization={Department of Computer Sciences, Faculty of Exact
    and Natural Sciences, Iv. Javakhishvili Tbilisi State University},%Department and Organization
            addressline={1 Ilia Tchavtchavadze Avenue, Academic Building I}, 
            city={Tbilisi},
            postcode={0179}, 
            state={Tbilisi},
            country={Georgia}}

\affiliation{organization={Department of Applied Physics II, Polytechnic
    School, University of Malaga},%Department and Organization
            addressline={c/ Dr. Ortiz Ramos s/n}, 
            city={Malaga},
            postcode={29071}, 
            state={Andalucía},
            country={Spain}}
  %
%  \email{jmss@uma.es [J.M. Sánchez]; nana\_georgiana@yahoo.com
%    [N. Odishelidze]; fcriado@uma.es [F. Criado-Aldenaueva]}}

\begin{abstract}
  We characterize Nash equilibria in $N$-player non-cooperative games
  governed by strict unidirectional local dependencies and binary
  strategies. Addressing the PPAD-completeness of general network
  games, we exploit the dyadic modularity of the cycle to derive a
  deterministic, closed-form solution in linear time ($O(N)$). We
  prove that active non-zero boundary incentives effectively linearize
  the topology into a feed-forward deterministic propagation. In
  robust regimes, strict dominance guarantees a unique equilibrium; in
  its absence, the existence of pure strategy equilibria is strictly
  governed by the \textit{Parity Condition} of differentiating
  agents—yielding exactly two or zero profiles—while a unique fully
  mixed equilibrium is universally guaranteed via induced payoff
  indifference. Furthermore, we extend this framework to topological
  singularities, demonstrating that structural indifference decouples
  the cycle into tractable linear segments, while boundary
  indifference requires a bifurcation analysis to resolve conditional
  dependencies. This analytical transparency enables the inverse
  design of inventory policies in circular supply chains, identifying
  the precise local parameters required to stabilize global
  synchronization.
\end{abstract}

\begin{highlights}
\item Unidirectional dependencies reduce equilibrium search to linear
  time $O(N)$.
\item Topological decouplers linearize the cycle into deterministic
  segments.
\item A hierarchical algorithm resolves equilibria across all
  incentive regimes.
\item Nodes are characterized as entropic sources or deterministic
  propagators.
\item The framework enables inverse analysis of inventory policies in
  supply chains.
\end{highlights}

\end{frontmatter}

\section{Introduction}
\label{sec:intro}

Sequential decision propagation governs the dynamics of critical
structures in Operations Research, from multi-echelon supply chains
\cite{CachonNetessine2004} to decentralized manufacturing rings
\cite{li2025strategic}. In these systems, stability is strictly a
function of local unilateral dependencies between adjacent
agents. While network game theory has characterized equilibrium
conditions under general topologies using spectral properties of
adjacency matrices \cite{acemoglu2015networks} or assumptions of
supermodularity \cite{goyal2007connections}, these approaches often
fail to provide operational insights for non-potential,
non-cooperative interactions.

A fundamental barrier in this domain is computational complexity. It
is established that computing Nash equilibria in general $N$-player
games is PPAD-complete \cite{Chen2009,Daskalakis2009}. This
intractability precludes the application of generic solvers to
large-scale operational networks, as they fail to scale polynomially
with system size. Furthermore, numerical methods (e.g., Lemke-Howson
or gradient descent) treat the system as a ``black box,'' generating
equilibrium vectors without revealing the structural causality
required for system optimization.

We address this challenge by developing a rigorous analytical
framework for $N$-player binary games with strict unidirectional
dependencies. We prove that the specific modularity of the dyadic
cycle collapses the complexity class of the problem, allowing for a
reduction from exponential to a deterministic linear-time solution
($O(N)$). Unlike models relying on probabilistic stability or
asymptotic learning \cite{bala1998learning}, we focus on the static
structural properties of Nash equilibria, exploiting the linearity of
local payoff functions to derive deterministic propagation rules.

Our closed-form characterization opens the ``white box'' of
causality. By explicitly mapping how local incentives and indifference
propagate, our framework enables \emph{Inverse System Design}:
identifying exactly which local parameters must be tuned to engineer a
specific global outcome—such as synchronizing inventory waves in a
circular supply chain—a task that remains opaque to spectral or
gradient-based numerical methods.

Our contribution is threefold:
\begin{enumerate}
\item We establish the mechanism of \emph{Deterministic
Propagation}. We prove that if a player's pure strategy encounters a
  non-zero marginal incentive in the subsequent player (active
  boundary condition), it triggers a deterministic chain
  reaction. This mechanism effectively linearizes the cyclic topology
  into a sequential dependency path, eliminating strategic uncertainty
  along the propagation chain.
    
\item We provide a complete characterization of equilibria in
  \emph{Robust Systems} (defined by strictly non-zero boundary
  incentives). We prove that:
  \begin{itemize}
  \item The presence of a dominant strategy guarantees a \emph{unique}
    Nash equilibrium, consisting entirely of pure strategies.
  \item In the absence of dominance, the existence of pure strategy
    equilibria is strictly governed by the \emph{Parity Condition} of
    differentiating players: even cycles support exactly two pure
    equilibria, while odd cycles support none.
  \item Regardless of parity, these regimes universally guarantee the
    existence of a \emph{unique} fully mixed strategy equilibrium.
  \end{itemize}
  
\item We develop a hierarchical taxonomy for \emph{Topological
Singularities}. We identify \emph{Dominance} and \emph{Structural
Indifference} as unconditional decouplers that resolve the cycle into
  linear segments. Furthermore, for cases of \emph{Boundary
  Indifference}, we introduce a bifurcation algorithm that resolves
  conditional dependencies. This perspective categorizes the
  equilibrium search into deterministic propagation and hypothesis
  verification, facilitating the algorithmic resolution of non-robust
  systems.
\end{enumerate}

\bigskip

The remainder of the paper is structured as follows: Section
\ref{sec:model} formulates the model. Sections \ref{sec:incentives}
through \ref{sec:mixed_strategies} derive the theoretical conditions
for pure and mixed strategies. Section \ref{sec:taxonomy} synthesizes
the taxonomy and the resolution algorithm. Section
\ref{sec:applications} validates the framework through Supply Chain
applications, demonstrating the inverse design of inventory policies,
before concluding in Section \ref{sec:conclusions}.

\section{Model Formulation}\label{sec:model}

We define a finite non-cooperative game in normal form
\begin{equation*}
  \Gamma=\langle I, \{S_i\}_{i\in I},\{H_i\}_{i\in I} \rangle.
\end{equation*}
The set of players is denoted by $I=\{1,\ldots,N\}$, indexed modulo
$N$ (implying $0 \equiv N$ and $N+1 \equiv 1$). We restrict our
analysis to binary strategy spaces, where each player $i$ selects a
pure strategy $s_i \in S_i=\{0,1\}$. The global strategy profile is
denoted by $s = (s_1, \dots, s_N) \in S = \prod_{i\in I} S_i$.

The interaction topology is strictly governed by a directed cycle
graph $C_N$ (Figure \ref{fig:topology}). This imposes a
\textit{unidirectional dependency constraint}: the payoff of player
$i$ is a function strictly of their own strategy $s_i$ and the
strategy of the immediate predecessor $s_{i-1}$.

\begin{figure}[htbp]
  \centering
  \begin{tikzpicture}[
      scale=0.6,
      player/.style={circle, draw=black, thick, fill=white, minimum size=1cm, inner sep=0pt, font=\bfseries},
      dependency/.style={->, >={Latex[length=3mm, width=2mm]}, thick, shorten >=2pt, shorten <=2pt},
      label text/.style={font=\footnotesize, align=center}
    ]
    
    % Definir radio y nodos principales
    \def\R{3.5cm}
    
    % Nodos estratégicos
    \node[player] (i) at (180:\R) {$i$};
    \node[player] (im1) at (240:\R) {$i-1$};
    \node[player] (ip1) at (110:\R) {$i+1$};

    %\draw[dashed] (0,0) circle (\R);
    
    % Nodos de contexto (puntos suspensivos)
    \node[player, dashed, draw=gray] (n) at (0:\R) {$N$};
    \node[player, dashed, draw=gray] (1) at (300:\R) {$1$};
    
    % Arcos de dependencia (Sentido horario o antihorario según tu convención modulo n)
    % Asumiendo i responde a i-1 (flujo de información)
    \draw[dependency] (im1) to[bend left=20] node[midway, left] {$s_{i-1}$} (i);
    \draw[dependency] (i) to[bend left=20] node[midway, left] {$s_{i}$} (ip1);
    \draw[dependency] (1) to[bend left=20] (im1);
    \draw[dependency] (n) to[bend left=20] (1);
    
    % Puntos suspensivos para completar el ciclo
    \draw[dashed, gray, thick] (ip1) to[bend left=15] (60:\R);
    \draw[dashed, gray, thick] (60:\R) to[bend left=20] (n);

    % Anotación de Función de Pago (La clave del modelo)
    \node[left=0cm of i, text width=3.5cm, font=\small] (payoff) {
        \textbf{Local Payoff:}\\
        $H_i(s_{i-1}, s_i)$ \\
        \textit{Dependence is strictly unidirectional}
    };
    \draw[thin, gray] (payoff.east) -- (i.west);

    % Anotación de Estrategia Binaria
    %\node[above=0.5cm of i, font=\footnotesize] {Choosing $s_i \in \{0,1\}$};
  \end{tikzpicture}
  \caption{\textbf{Directed Cycle Topology.} The diagram illustrates
    the $N$-player game structure where player $i$'s payoff is
    strictly determined by their own strategy and the predecessor's
    strategy $s_{i-1}$. This unidirectional coupling allows for
    sequential propagation analysis.}
\label{fig:topology}
\end{figure}
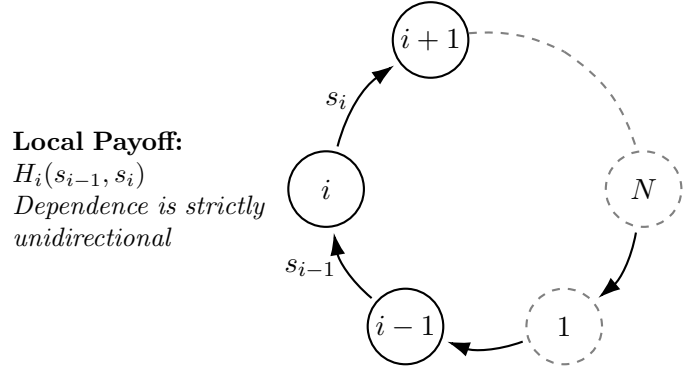

\begin{definition}[Local Payoff Structure]
  The global payoff function $H_i: S \to \mathbb{R}$ is structurally
  restricted to the local mapping $H_i(s_{i-1}, s_i)$. The strategic
  incentives for player $i$ are fully characterized by the payoff
  matrix $\mathbf{H}_i \in \mathbb{R}^{2\times 2}$:
  \begin{equation}
    \mathbf{H}_i = \begin{pmatrix}
      H_i(0, 0) & H_i(0, 1)\\
      H_i(1, 0) & H_i(1, 1)
    \end{pmatrix}
    \label{eq:payoff_matrix}
  \end{equation}
  where rows correspond to the input state $s_{i-1} \in \{0,1\}$ and
  columns to the decision $s_i \in \{0,1\}$.
\end{definition}

This formulation unifies cyclic and linear topologies. A
\textit{Chain} structure is mathematically formalized as a Cycle where
player 1 exhibits \textit{Structural Indifference} to player $n$
(i.e., $H_1(0, s_1) = H_1(1, s_1)$ for all $s_1$). In this case,
player 1 acts as an exogenous source node.

A pure strategy profile $s^* \in S$ constitutes a Nash equilibrium
\cite{nash1951} if no player can improve their payoff by unilateral
deviation:
\begin{equation}
  H_i(s^*_{i-1},s^*_i) \geq H_i(s^*_{i-1}, s_i), \quad \forall s_i\in S_i, \quad \forall i\in I.
  \label{eq:pure_nash}
\end{equation}

We extend the analysis to mixed strategies $\sigma_i = (1-p_i, p_i)$,
where $p_i \in [0,1]$ denotes the probability of selecting strategy
1. By the linearity of expectation, the payoff function $H_i(p_{i-1},
p_i)$ takes the bilinear form:
\begin{equation}
  H_i(p_{i-1}, p_i) = 
  \begin{pmatrix} 1 - p_{i-1} & p_{i-1} \end{pmatrix} 
  \mathbf{H}_i 
  \begin{pmatrix} 1 - p_i \\ p_i \end{pmatrix}.
    \label{eq:mixed_payoff}
\end{equation}
A critical structural distinction arises here: unlike general $N$-player
games, where expected payoffs entail multilinear polynomials of degree
$N$, the unidirectional cycle restricts interactions to be strictly
pairwise. Consequently, the marginal utility for player $i$ is an
affine function of the single variable $p_{i-1}$, completely
decoupled from the strategies of players $j \notin \{i-1, i\}$.

\section{Strategic Incentives and Local Payoff Structure}
\label{sec:incentives}

\subsection{The Marginal Incentive Function}

The bilinear structure of the expected payoff $H_i(p_{i-1}, p_i)$
allows for a canonical decomposition separating the player's control
variable from the environmental constraint. We define the
\textit{Marginal Incentive Function} $\Delta_i: [0,1] \to \mathbb{R}$
as the net payoff gain from choosing strategy 1 over 0:
\begin{equation}
  \begin{aligned}
    H_i(p_{i-1}, p_i)
    &= (1-p_i)\cdot H_i(p_{i-1},0) + p_i \cdot H_i(p_{i-1},1)\\
    &=H_i(p_{i-1}, 0) + p_i \cdot \underbrace{\left[ H_i(p_{i-1}, 1) - H_i(p_{i-1}, 0) \right]}_{\Delta_i(p_{i-1})}.
  \end{aligned}
    \label{eq:payoff_decomposition}
\end{equation}
Since the expectation over the predecessor's strategy is linear,
$\Delta_i(p_{i-1})$ is strictly affine. It is fully determined by the
convex combination of the boundary incentives $\Delta_i(0)$ and
$\Delta_i(1)$:
\begin{equation}
  \begin{aligned}
    \Delta_i(p_{i-1})
    &= (1-p_{i-1})\Delta_i(0) + p_{i-1}\Delta_i(1).\\
    &= \Delta_i(0) + p_{i-1} (\Delta_i(1) - \Delta_i(0)).
  \end{aligned}
  \label{eq:incentive_affine}
\end{equation}
This affine property dictates that the incentive function $\Delta_i$
is monotonic: it cannot oscillate, being strictly restricted to
constant, strictly increasing, or strictly decreasing behaviors.

\subsection{Classification of Best Response Regimes}

The best response correspondence is determined solely by the sign of
$\Delta_i(p_{i-1})$. The linearity ensures that the roots of
$\Delta_i(p_{i-1}) = 0$ (points of indifference) are governed entirely
by the boundary values. We establish the following trichotomy:

\begin{lemma}[Incentive Characterization]
  \label{lem:incentive_regimes}
  For any player $i$, the strategic dependency on $i-1$ falls into one
  of three regimes:
  \begin{itemize}
  \item \emph{Dominance (Decoupling):} If $\Delta_i(0)$ and
    $\Delta_i(1)$ share the same non-zero sign (i.e.,
    $\Delta_i(0)\cdot\Delta_i(1) > 0$), then the sign of the marginal
    incentive, $\text{sgn}(\Delta_i(p_{i-1}))$, is invariant for all
    $p_{i-1} \in [0,1]$.
    
    \emph{Implication:} Player $i$ possesses a unique, strict best
    response independent of the predecessor's strategy $s_{i-1}$.

  \item \emph{Structural Indifference:} If $\Delta_i(0) = \Delta_i(1)
    = 0$, then $\Delta_i(p_{i-1}) \equiv 0$ for all $p_{i-1} \in
    [0,1]$.

    \emph{Implication:} Player $i$'s payoff is invariant to their own
    strategy choice, regardless of the predecessor's behavior.
  \item \emph{Switching Regime:} If the boundary incentives differ in
    sign or exactly one vanishes (i.e., $\Delta_i(0)\cdot\Delta_i(1)
    \leq 0$, excluding structural indifference), there exists a unique
    critical threshold $p^*_{i-1}$ defined by:
    \begin{equation}
      p^*_{i-1} = \frac{\Delta_i(0)}{\Delta_i(0) - \Delta_i(1)}.
      \label{eq:critical_threshold}
    \end{equation}
    The strategic implication is determined by the strictness of the
    boundary condition:
    \begin{itemize}
    \item \emph{Strict Sign Change ($\Delta_i(0)\cdot\Delta_i(1) <
    0$):} The threshold is strictly interior ($0 < p^*_{i-1} < 1$). It
      partitions the input space into two domains, $[0, p^*_{i-1})$
        and $(p^*_{i-1}, 1]$, where the optimal response of player $i$
      is strictly opposite. Payoff indifference is induced if and only
      if the predecessor plays exactly $p_{i-1} = p^*_{i-1}$.
    
    \item \emph{Boundary Indifference:} Exactly one boundary incentive
      vanishes. The threshold coincides with a vertex ($p^*_{i-1} \in
      \{0, 1\}$). Player $i$ is induced to indifference solely at that
      specific vertex, maintaining a strict preference throughout the
      remainder of the domain.
    \end{itemize}
  \end{itemize}
\end{lemma}

\begin{proof}
  The proof follows directly from the properties of affine functions
  on a compact interval. We analyze the equation $\Delta_i(p_{i-1}) =
  0$:
  \begin{enumerate}
  \item \emph{Dominance:} If $\Delta_i(0)$ and $\Delta_i(1)$ share the
    same sign, the convex hull of these boundary values excludes
    zero. Linearity ensures that the image of the interval $[0,1]$
    under $\Delta_i$ lies strictly to one side of the origin.

    Consequently, the sign is invariant, implying that the payoff
    $H_i(p_{i-1}, p_i)$ is strictly monotonic in $p_i$, achieving its
    maximum uniquely at a boundary.

  \item \textit{Structural Indifference:} Trivial. Since the function
    is affine and vanishes at both endpoints, it is identically zero
    on the entire domain.

    Thus, player $i$'s payoff becomes independent of their strategy
    selection $p_i$.

  \item \textit{Switching Regime:} The condition
    $\Delta_i(0)\cdot\Delta_i(1) \leq 0$ implies that zero lies within
    the range $[\min(\Delta_i), \max(\Delta_i)]$. Since $\Delta_i$ is
    affine and non-constant (excluding structural indifference), it
    constitutes a monotonic bijection from $[0,1]$ to its image. By
    the Intermediate Value Theorem, there exists a unique pre-image
    $p^*_{i-1} \in [0,1]$ such that $\Delta_i(p^*_{i-1})=0$. Solving
    for $p_{i-1}$ in Equation \eqref{eq:incentive_affine} yields the
    explicit form \eqref{eq:critical_threshold}. Specifically:
    \begin{itemize}
    \item If $\Delta_i(0)\cdot\Delta_i(1) < 0$, monotonicity ensures
      the root is strictly interior.

      The marginal incentive $\Delta_i(\cdot)$ vanishes at
      $p^*_{i-1}$ and assumes opposite non-zero signs on the
      sub-intervals $[0, p^*_{i-1})$ and $(p^*_{i-1}, 1]$, thereby
      inverting the optimal strategy $p_i$ across the threshold.
    \item If a boundary vanishes, the root coincides trivially with
      that boundary.

      Thus, if $p_{i-1}=p^*_{i-1}$, player $i$'s payoff is independent
      of $p_i$. For all inputs $p_{i-1} \neq p^*_{i-1}$, the optimal
      strategy remains constant.
    \end{itemize}
  \end{enumerate}
\end{proof}

This characterization establishes a fundamental operational insight:
despite the continuous nature of the mixed strategy space $[0,1]$, the
local incentive structure inherently polarizes behavior towards the
boundaries. Unless the predecessor aligns their strategy exactly with
the critical threshold $p^*_{i-1}$ (inducing indifference), player
$i$'s optimal response is deterministically driven to an extreme point
$p_i \in \{0,1\}$. This ``digital'' response to ``analog'' inputs
serves as the primary mechanism for the propagation of pure strategies
(Section \ref{sec:pure_strategies}) and necessitates precise inverse
engineering of probabilities to sustain mixed equilibria (Section
\ref{sec:mixed_strategies}).

\section{Optimal Local Strategy}
\label{sec:local_strategy}

We now translate the incentive structures derived in Lemma
\ref{lem:incentive_regimes} into the optimal decision rules for player
$i$. In this non-cooperative framework, each agent selects $p_i \in
    [0, 1]$ to maximize the expected payoff $H_i$, conditional on the
    predecessor's strategy $p_{i-1}$.

Recall that the objective function is affine in the decision variable
(Eq. \eqref{eq:payoff_decomposition}):
\begin{equation}
    H_i(p_{i-1}, p_i) = \text{Constant} + p_i \cdot \Delta_i(p_{i-1})
\end{equation}
The optimization problem is thus entirely governed by the sign of the
Marginal Incentive Function $\Delta_i(p_{i-1})$. This allows us to map
the incentive regimes identified in Lemma \ref{lem:incentive_regimes}
directly to the Best Response Correspondence $\mathcal{B}_i$.

\begin{proposition}[Local Best Response]
  \label{prop:best_response}
  The best response correspondence $\mathcal{B}_i: [0, 1]
  \rightrightarrows [0, 1]$ is determined by the sign of the marginal
  incentive $\Delta_i(p_{i-1})$:
  \begin{equation}
    \mathcal{B}_i(p_{i-1}) = 
    \begin{cases} 
      \{1\} & \text{if } \Delta_i(p_{i-1}) > 0, \\
      \{0\} & \text{if } \Delta_i(p_{i-1}) < 0, \\
            [0, 1] & \text{if } \Delta_i(p_{i-1}) = 0.
    \end{cases}
    \label{eq:best_response_map}
  \end{equation}
\end{proposition}

\begin{proof}
  The proof follows immediately from the linearity of $H_i(p_{i-1},
  \cdot)$. If the slope $\Delta_i(p_{i-1})$ is non-zero, the maximum
  on the compact interval $[0,1]$ is unique and located at the
  boundary corresponding to the slope's direction. If the slope is
  zero, the payoff is constant, implying the player indifferent across
  the entire domain.
\end{proof}

\subsection{Operational Implications: The Binary Default vs. Indifference}

This proposition establishes that equilibrium analysis in binary games
with unidirectional dependencies is fundamentally discrete, admitting
mixed strategies only under specific singular conditions. We classify
the operational behavior of node $i$ based on the interplay between
the predecessor's input and the local incentive structure (Lemma
\ref{lem:incentive_regimes}):
\begin{enumerate}
\item \emph{Deterministic Response (The Binary Default):} Whenever
  $\Delta_i(p_{i-1}) \neq 0$, the optimal strategy is strictly pure (a
  ``bang-bang'' solution). This occurs in two scenarios:
  \begin{itemize}
  \item \emph{Dominance:} If player $i$ operates in a Dominance Regime
    (Lemma \ref{lem:incentive_regimes}), the sign of $\Delta_i$ is
    invariant. The player adopts a fixed pure strategy regardless of
    $p_{i-1}$, effectively \emph{decoupling} from the predecessor.
  \item \emph{Active Propagation:} If player $i$ is in a Switching
    Regime but the input $p_{i-1}$ differs from the critical value
    $p^*_{i-1}$, the player is forced to a boundary. This mechanism
    drives the \emph{propagation} of pure strategies analyzed in
    Section \ref{sec:pure_strategies}.
    \end{itemize}

\item \emph{The Indifference Condition (Mixing Gateway):} Strictly
  mixed strategies ($p_i \in (0, 1)$) are admissible if and only if
  the marginal incentive vanishes ($\Delta_i(p_{i-1}) = 0$). This
  condition induces local indifference, rendering any strategy in
  $[0,1]$ optimal. This state arises via:
  \begin{itemize}
  \item \emph{Structural Indifference:} The node is inherently passive
    ($\Delta_i \equiv 0$).
  \item \emph{Induced Indifference:} The predecessor actively tunes
    $p_{i-1}$ to the critical threshold $p^*_{i-1}$ (Lemma
    \ref{lem:incentive_regimes}) to nullify player $i$'s incentive.
  \end{itemize}
\end{enumerate}

This dichotomy governs the subsequent analysis. Section
\ref{sec:pure_strategies} characterizes the \emph{propagation} of
deterministic behaviors, while Section \ref{sec:mixed_strategies}
formalizes the precise construction of \emph{Induced Indifference}
required to sustain mixed equilibria.

\section{Deterministic Propagation and Cycle Stability}
\label{sec:pure_strategies}

The computation of Nash equilibria in general $N$-player games is
known to be PPAD-complete \cite{Chen2009,Daskalakis2009}, and
determining the existence of pure strategy equilibria is typically
NP-hard \cite{Gottlob2005}. By contrast, the unidirectional ring
topology constrains the dependency graph, allowing us to decouple the
complex multivariate optimization into a sequential propagation
mechanism solvable in linear time.

We first establish that mixed strategies constitute degenerate
phenomena in this framework. By virtue of the linearity of the payoff
functions, strictly mixed strategies cannot exist as robust optima;
they emerge exclusively as a consequence of precise marginal
indifference.

\begin{lemma}[Structure of Equilibrium Responses]
  \label{lem:response_structure}
  Let $p^* = (p_1^*, \dots, p_N^*)$ be a Nash equilibrium profile,
  where $p_i^* \in [0,1]$ denotes the probability of selecting
  strategy 1. For every player $i$, the following disjunction holds:
  \begin{equation}
    p_i^* \in \{0, 1\} \quad \lor \quad \Delta_i(p_{i-1}^*) = 0, \qquad \forall i\in I
  \end{equation}
  Operationally, a player acts deterministically unless the marginal
  incentive vanishes. Consequently, strictly mixed strategies ($p_i^*
  \in (0, 1)$) are admissible exclusively under strict indifference.
\end{lemma}

\begin{proof}
  Recalling the payoff decomposition
  (Eq. \ref{eq:payoff_decomposition}), the optimization problem for
  player $i$ reduces to maximizing a linear function of $p_i$:
  \begin{equation*}
    \max_{p_i \in [0,1]} \quad p_i \cdot \Delta_i(p_{i-1}^*) + \text{Constant}.
  \end{equation*}
  \begin{itemize}
  \item If $\Delta_i(p_{i-1}^*) > 0$, the slope is positive; the
    unique maximum is $p_i^* = 1$.
  \item If $\Delta_i(p_{i-1}^*) < 0$, the slope is negative; the
    unique maximum is $p_i^* = 0$.
  \item If $\Delta_i(p_{i-1}^*) = 0$, the objective function is
    constant, rendering any strategy $p_i^* \in [0, 1]$ optimal.
  \end{itemize}
  Thus, $p_i^*$ can deviate from the boundary set $\{0, 1\}$ (i.e., be
  strictly mixed) if and only if the marginal incentive term vanishes.
\end{proof}

\section{Deterministic Propagation and Cycle Stability}
\label{sec:pure_strategies}

The computation of Nash equilibria in general $N$-player games is
PPAD-complete \cite{Daskalakis2009,Chen2009}, and determining the
existence of pure strategies is NP-hard \cite{Gottlob2005}. In
contrast, the unidirectional ring topology structurally collapses this
complexity, decoupling the global multivariate optimization into a
linear sequential decision chain.

We construct the equilibrium set by tracing the forward propagation of
deterministic constraints.

\subsection{Local Incentives and Propagation}

We first characterize the admissible states of any player within a
Nash equilibrium. Building on the linearity established in Lemma
\ref{lem:response_structure}, strictly mixed strategies cannot exist
as robust optima; they appear exclusively as a consequence of
vanishing marginal incentives.

\begin{lemma}[Regime Dichotomy]
  \label{lem:dichotomy}
  Let $p^*$ be a Nash equilibrium profile. For any player $i$,
  the optimality condition imposes exactly one of the following two
  mutually exclusive regimes:
  \begin{enumerate}
  \item \emph{Strict Preference (Deterministic):} The player is \emph{compelled}
    to select a unique pure strategy $p_i^* \in \{0, 1\}$. 
  \item \emph{Indifference:} The player is \emph{unconstrained}
    regarding the strategy choice, rendering any $p_i^* \in [0, 1]$ optimal.
  \end{enumerate}
\end{lemma}

\begin{proof}
  The optimization is linear in the decision variable $p_i$. If the
  marginal incentive $\Delta_i(p^*_{i-1})$ is non-zero, the maximum is
  unique and located at a boundary. If the marginal incentive vanishes,
  the gradient is zero, and the domain of optimality encompasses the
  entire interval $[0,1]$.
\end{proof}

This dichotomy governs the sequential transmission of constraints
through the ring. If a player adopts a pure strategy, they resolve the
strategic uncertainty for the successor, potentially triggering a
deterministic cascade.

\begin{proposition}[Deterministic Propagation]
  \label{prop:propagation}
  Suppose player $i-1$ plays a pure strategy $p_{i-1} \in \{0, 1\}$
  such that the induced marginal incentive for player $i$ is non-zero
  (i.e., $\Delta_i(p_{i-1}) \neq 0$). Then, player $i$ is compelled
  to adopt a unique pure strategy $p_i \in \{0, 1\}$.
\end{proposition}

\begin{proof}
  Since $p_{i-1} \in \{0,1\}$, the marginal incentive function
  evaluates exactly to one of the boundary values, $\Delta_i(0)$ or
  $\Delta_i(1)$. By hypothesis, this value is non-zero. Applying Lemma
  \ref{lem:dichotomy}, the Indifference regime is
  precluded. Consequently, the unique optimal response is the pure
  strategy $p_i \in \{0, 1\}$ determined by the sign of
  $\Delta_i(p_{i-1})$.
\end{proof}

This dichotomy governs the sequential transmission of constraints
through the ring. If a player adopts a pure strategy, they resolve the
strategic uncertainty for the successor, potentially triggering a
deterministic cascade.

\begin{proposition}[Deterministic Propagation]
  \label{prop:propagation}
  Suppose player $i-1$ adopts a pure strategy $p_{i-1} \in \{0, 1\}$
  such that the induced marginal incentive for player $i$ is non-zero
  (i.e., $\Delta_i(p_{i-1}) \neq 0$). Then, player $i$ is compelled
  to adopt a unique pure strategy $p_i \in \{0, 1\}$.
\end{proposition}

\begin{proof}
  Since $p_{i-1} \in \{0,1\}$, the marginal incentive function
  evaluates exactly to one of the boundary values, $\Delta_i(0)$ or
  $\Delta_i(1)$. By hypothesis, this value is non-zero. Applying Lemma
  \ref{lem:dichotomy}, the Indifference regime is
  precluded. Consequently, the unique optimal response is the pure
  strategy $p_i \in \{0, 1\}$ determined by the sign of
  $\Delta_i(p_{i-1})$.
\end{proof}

\subsection{Global Stability under Robust Incentives}

To characterize global equilibria, we must first isolate systems where
determinism is not lost at the boundaries due to incidental zero
gradients. We define a structural condition that precludes boundary
indifference.

\begin{definition}[Robust Incentive Structure]
  \label{def:robust_structure}
  A game possesses a \emph{Robust Incentive Structure} if, for every
  player $i \in I$, the marginal incentives at the boundaries are
  strictly non-zero:
  \begin{equation}
    \Delta_i(0) \cdot \Delta_i(1) \neq 0.
    \label{eq:robust_condition}
  \end{equation}
\end{definition}

Under this condition, the system segregates into two mutually
exclusive regimes.

\begin{theorem}[System-Wide Dichotomy]
  \label{thm:system_purity}
  In a game exhibiting a Robust Incentive Structure, the set of Nash
  equilibria is partitioned into two disjoint classes:
  \begin{enumerate}
  \item \emph{Total Purity:} Every player adopts a pure strategy ($p_i
    \in \{0, 1\}$, $\forall i\in I$).
  \item \emph{Total Indifference:} Every player adopts a
    \emph{strictly mixed} strategy ($p_i \in (0, 1)$, $\forall i\in
    I$), exhibiting payoff indifference.
  \end{enumerate}
  Hybrid profiles (coexistence of pure and mixed strategies) are
  structurally impossible.
\end{theorem}

\begin{proof}
  The proof proceeds by contradiction. Assume a hybrid equilibrium
  exists, containing at least one player $k$ employing a pure strategy
  ($p_k \in \{0, 1\}$).
  
  By the \emph{Robust Incentive Structure} (Definition
  \ref{def:robust_structure}), the boundary incentives are strictly
  non-zero ($\Delta_{k+1}(0) \neq 0$ and $\Delta_{k+1}(1) \neq
  0$). Consequently, the pure choice of player $k$ induces a non-zero
  marginal incentive for player $k+1$. By Proposition
  \ref{prop:propagation}, player $k+1$ is compelled to adopt a unique
  pure strategy.
  
  By finite induction, this deterministic constraint propagates
  sequentially through the cycle ($k \to k+1 \to \dots \to N \to 1\to
  k-1$), forcing every player to adopt a pure strategy. This
  contradicts the assumption of a hybrid profile.
  
  Therefore, the only alternative to a pure profile is one where
  \emph{no} player plays a pure strategy, implying $p_i \in (0, 1)$
  for all $i \in I$. By Lemma \ref{lem:dichotomy}, strict mixing
  necessitates indifference, confirming the second case.
\end{proof}

\subsection{Characterization of Pure Equilibria}

We now focus on the existence and uniqueness of \textit{Total Purity}
equilibria. We classify local player behavior based on the response to
the predecessor's pure strategy.

\begin{definition}[Strategic Taxonomy]
  \label{def:taxonomy}
  Under Robust Incentives, player $i$ belongs to one of three classes:
  \begin{enumerate}
  \item \emph{Dominant:} The optimal strategy $p_i$ is invariant to
    the predecessor's state $p_{i-1}$.
    \begin{equation*}
      \Delta_i(0)\cdot \Delta_i(1) > 0
    \end{equation*}
  \item \emph{Conformist:} The player replicates the predecessor's
    strategy ($0 \to 0, 1 \to 1$).
    \begin{equation*}
      \Delta_i(0) < 0 \quad \text{and} \quad \Delta_i(1) > 0
    \end{equation*}
  \item \emph{Inverter:} The player negates the predecessor's
    strategy ($0 \to 1, 1 \to 0$).
    \begin{equation*}
      \Delta_i(0) > 0 \quad \text{and} \quad \Delta_i(1) < 0
    \end{equation*}
  \end{enumerate}
\end{definition}

This classification yields a complete topological resolution of the game.

\begin{theorem}[Uniqueness via Dominance]
  \label{thm:dominance_uniqueness}
  Under Robust Incentives, if there exists at least one
  \emph{Dominant} player, the game admits a \emph{unique} Nash
  equilibrium, which is strictly pure.
\end{theorem}

\begin{proof}
  Let player $k$ be a Dominant node. Their optimal strategy $p_k^* \in
  \{0,1\}$ is invariant. Since incentives are robust, $p_k^*$ uniquely
  determines $p_{k+1}^*$ via Proposition \ref{prop:propagation}. This
  determination cascades through the ring ($k \to k+1 \to \dots \to
  k-1$). Cycle consistency is \emph{ensured} by definition: since
  player $k$'s choice is unconditional, the equilibrium condition
  holds trivially regardless of the feedback from $k-1$, preventing
  any logical frustration.
\end{proof}

If no player is Dominant, the system operates as a closed feedback
loop composed exclusively of Conformists and Inverters.

\begin{theorem}[Parity Condition]
  \label{thm:parity_condition}
  Consider a game with Robust Incentives where no player is
  Dominant. Let $\kappa$ denote the number of \emph{Inverter} players
  in the cycle.
  \begin{itemize}
  \item \emph{Even Parity ($\kappa$ is even):} There exist exactly
    \emph{two} pure strategy Nash equilibria.
  \item \emph{Odd Parity ($\kappa$ is odd):} There exists \emph{no}
    pure strategy Nash equilibrium.
  \end{itemize}
\end{theorem}

\begin{proof}
  Since incentives are robust and dominance is absent, the best
  response function $B_i: \{0,1\} \to \{0,1\}$ is a bijection for
  every player:
  \begin{equation*}
    B_i(x) = \begin{cases}
      x & \text{if } i \text{ is Conformist} \\
      1-x & \text{if } i \text{ is Inverter}
    \end{cases}
  \end{equation*}
  
  We define the \emph{Cycle Return Map} $\Phi: \{0,1\} \to \{0,1\}$ as
  the composition of all local responses around the ring: $\Phi(x) =
  B_N(\dots B_2(B_1(x))\dots)$. A pure strategy equilibrium
  corresponds strictly to a fixed point $x^* = \Phi(x^*)$.
  
  Since each Inverter introduces a logical negation, the nature of $\Phi(x)$ depends on the
  parity of $\kappa$:
  \begin{itemize}
  \item \emph{Case $\kappa$ is Even:} The negations cancel out,
    yielding the identity map $\Phi(x) = x$. This fixed-point equation holds for
    both $x=0$ and $x=1$. Thus, seeding the cycle with $p_1=0$ and $p_1=1$ generates
    two distinct valid equilibria.
  \item \emph{Case $\kappa$ is Odd:} The net effect is a negation,
    yielding $\Phi(x) = 1-x$. The equation $x = 1-x$ has no solution in the
    binary set $\{0,1\}$. Thus, logical consistency is impossible (Cycle Frustration), and
    no pure equilibrium exists.
  \end{itemize}
\end{proof}

\section{Induced Indifference and Backward Propagation}
\label{sec:mixed_strategies}

While Section \ref{sec:pure_strategies} established that pure strategy
equilibria rely on the \emph{forward propagation} of deterministic
constraints ($i \to i+1$), the existence of mixed strategies requires
the \emph{neutralization} of these constraints.

In a binary framework, non-zero boundary incentives
($\Delta_i(p_{i-1}) \neq 0$) compel a boundary choice, precluding
mixing. Consequently, a player $i$ adopts a strictly mixed strategy
$p_i \in (0,1)$ if and only if the marginal incentive vanishes. We
distinguish two topological mechanisms that generate this state,
yielding fundamentally different causal implications:
\begin{enumerate}
\item \emph{Structural Indifference ($\Delta_i(0) = \Delta_i(1) =
0$):} The node exhibits invariant incentives regardless of the
  input. This creates a \textit{Topological Decoupling}: player $i$
  acts as an informational sink, absorbing the predecessor's strategy
  with no effect on local optimality. Consequently, this indifference
  imposes no constraint on player $i-1$.

\item \emph{Induced Indifference (Active Cancellation):} The player
  possesses strict preferences but is neutralized by the
  predecessor. Here, the condition $\Delta_i(p_{i-1}) = 0$ imposes a
  \textit{binding constraint} on player $i-1$. This triggers
  \textit{Backward Propagation}, where the necessity for mixing at
  node $i$ strictly dictates the strategy of node $i-1$.
\end{enumerate}

\subsection{The Geometry of Induced Indifference}

We first analyze the case of active cancellation. Absent structural
indifference, the condition for player $i$ to mix necessitates a
precise strategic input from the predecessor.

\begin{lemma}[The Critical Predecessor Strategy]
  \label{lem:induced_indifference}
  Assume player $i$ exhibits non-constant marginal incentives
  ($\Delta_i(0) \neq \Delta_i(1)$). A feasible critical strategy
  $p^*_{i-1} \in [0,1]$ that renders player $i$ indifferent exists if
  and only if the boundary incentives have non-matching signs
  ($\Delta_i(0) \cdot \Delta_i(1) \le 0$). This strategy is unique and
  defined by:
  \begin{equation}
    p^*_{i-1} = \frac{\Delta_i(0)}{\Delta_i(0) - \Delta_i(1)}.
    \label{eq:critical_p}
  \end{equation}
  Furthermore, $p^*_{i-1}$ represents a strictly mixed strategy
  ($p^*_{i-1} \in (0,1)$) if and only if the inequality is strict
  ($\Delta_i(0) \cdot \Delta_i(1) < 0$).
\end{lemma}

\begin{proof}
  The proof follows directly from the affine properties of the
  marginal incentive established in Eq. (\ref{eq:incentive_affine})
  and the classification in Lemma \ref{lem:incentive_regimes}.
  
  By the assumption that player $i$ is not structurally indifferent,
  $\Delta_i$ is a non-constant affine function. The critical strategy
  $p^*_{i-1}$ corresponds to the unique root of $\Delta_i(p_{i-1})=0$,
  explicitly derived in Eq. (\ref{eq:critical_threshold}).
  
  The feasibility of this probability is governed by the boundary
  values:
  \begin{itemize}
  \item For $p^*_{i-1}$ to exist within the closed unit interval
    $[0,1]$, zero must lie between the boundary incentives
    $\Delta_i(0)$ and $\Delta_i(1)$ (by the Intermediate Value Theorem
    applied to an affine map). This requires $\Delta_i(0) \cdot
    \Delta_i(1) \le 0$.
  \item For $p^*_{i-1}$ to constitute a \textit{strictly mixed}
    strategy ($p^*_{i-1} \in (0,1)$), the root must be strictly
    interior. This is satisfied if and only if the boundaries possess
    strictly opposite signs, i.e., $\Delta_i(0) \cdot \Delta_i(1) <
    0$.
  \end{itemize}
\end{proof}

\subsection{Fully Mixed Nash Equilibria in Robust Systems}

We restrict our analysis to Robust Systems where no player possesses a
dominant strategy. Theorem \ref{thm:system_purity} established that
under these conditions, any Nash equilibrium is either fully pure or
fully mixed. Having characterized the pure strategy profiles in
Section \ref{sec:pure_strategies}, we now focus on the existence and
uniqueness of the fully mixed equilibrium.

\begin{theorem}[Uniqueness of the Fully Mixed Equilibrium]
  \label{thm:mixed_equilibrium}
  Consider a game with a Robust Incentive Structure and no dominant strategies:
  \begin{equation}
    \Delta_i(0) \cdot \Delta_i(1) < 0, \qquad \forall i\in I.
  \end{equation}
  There exists a \emph{unique} Nash equilibrium in strictly mixed
  strategies $p^* = (p_1^*, \dots, p_N^*)$, where $p_i^* \in
  (0,1)$ for all $i$. This equilibrium is constructed
  deterministically via a backward constraint:
  \begin{equation}
    p_i^* = \frac{\Delta_{i+1}(0)}{\Delta_{i+1}(0) - \Delta_{i+1}(1)}, \quad \forall i \in \{1, \dots, N\}
    \label{eq:mixed_algo}
  \end{equation}
  where indices are taken modulo $N$ (identifying $N+1 \equiv 1$).
\end{theorem}

\begin{proof}
  By Theorem \ref{thm:system_purity}, a fully mixed equilibrium $p^*$
  requires every player to exhibit payoff indifference. Since the
  Robust Incentive assumption precludes Structural Indifference
  ($\Delta_i \not\equiv 0$), this indifference must be actively
  induced via the mechanism established in Lemma
  \ref{lem:induced_indifference}.
  
  For any player $i+1$ to be indifferent, the predecessor $i$ is
  uniquely constrained to play the critical strategy $p_i^*$. The
  absence of dominance implies a Global Switching Regime
  ($\Delta_{i+1}(0) \cdot \Delta_{i+1}(1) < 0$ for all $i$). By Lemma
  \ref{lem:induced_indifference}, this strict sign change guarantees
  that the critical strategy lies strictly within the unit interval
  ($p_i^* \in (0,1)$).
  
  Since each $p_i^*$ is determined independently by the parameters of
  player $i+1$ (Eq. \ref{eq:mixed_algo}), the system of equations is
  decoupled, and the equilibrium profile is unique.
\end{proof}

\begin{remark}[Coexistence and Stability]
  It is crucial to note that Theorem \ref{thm:mixed_equilibrium}
  asserts the uniqueness of the \textit{mixed} equilibrium, not the
  uniqueness of the equilibrium \textit{set}. As shown in Theorem
  \ref{thm:parity_condition} (Parity Condition), if the number of
  inverters is even, this unique mixed equilibrium coexists with two
  pure strategy equilibria. In such cases, the mixed equilibrium
  represents a precarious "balancing act" of payoff indifference,
  whereas the pure equilibria represent stable waves of deterministic
  propagation.
\end{remark}

\section{Taxonomy of Equilibrium Structures and Complexity Reduction}
\label{sec:taxonomy}

Synthesizing the analytical results from Sections
\ref{sec:pure_strategies} and \ref{sec:mixed_strategies}, we establish
a unified taxonomy for $N$-player binary games. We demonstrate that
the interplay between local boundary incentives ($\Delta_i(0),
\Delta_i(1)$) and global topology (specifically, the verification of
\emph{cycle consistency} to ensure the absence of logical frustration,
formalized as \emph{Cycle Parity} strictly within robust regimes)
governs the cardinality and stability of the equilibrium
set. Furthermore, for systems exhibiting structural singularities
(deviations from robustness), the framework provides a mechanism to
decompose the cycle into tractable linear segments, preserving $O(N)$
solvability.

This classification serves a dual purpose: it provides a rigorous map
of the solution space and formally establishes the reduction of
computational complexity from PPAD-complete to deterministic linear
time.

\subsection{Regime I: The Robust Solvability Matrix}
In Robust Systems (Definition \ref{def:robust_structure}),
characterized by non-zero boundary incentives
($\Delta_i(0)\cdot\Delta_i(1) \neq 0$), the analysis bifurcates based
on the presence of Dominance and the Parity of Inverters. Table
\ref{tab:equilibrium_taxonomy} summarizes the exhaustive
classification established in Theorems \ref{thm:dominance_uniqueness},
\ref{thm:parity_condition}, and \ref{thm:mixed_equilibrium}:
\begin{itemize}
\item \emph{The Deterministic Anchor (Case A):} If at least one player
  possesses a dominant strategy, the feedback loop is functionally
  decoupled. This node acts as a \textit{Deterministic Source},
  compelling the subsequent $N-1$ players into a unique pure
  trajectory via Proposition \ref{prop:propagation}. Mixed strategies
  are impossible because the source unconditionally emits a boundary
  strategy, never inducing indifference in the successor.
  
\item \emph{The Parity Switch (Case B):} In the absence of dominance,
  the existence of pure equilibria is strictly determined by the
  logical consistency of the sequence of local responses. Even parity
  ensures a consistent return map, supporting two distinct pure
  strategy waves (Theorem \ref{thm:parity_condition}). Odd parity
  induces \emph{Cycle Frustration}, a logical contradiction that
  precludes pure strategies, forcing the system into the unique mixed
  equilibrium characterized by the induced payoff indifference
  described in Eq. \eqref{eq:mixed_algo}.
\end{itemize}

\begin{table}[htbp]
  \centering
  \caption{Global Solvability Matrix for Robust Cyclic Games. The
    solution structure is fully determined by local dominance and the
    global parity of differentiating agents ($\kappa$).}
  \label{tab:equilibrium_taxonomy}
  \small
  \begin{tabular}{p{.28\textwidth}p{.27\textwidth}p{.11\textwidth}p{.11\textwidth}}
    \toprule
    \emph{Structural Condition} & \emph{Topological Effect} & \multicolumn{2}{c}{\emph{Equilibrium Set Cardinality}} \\ \cmidrule(l){3-4} 
    &  & \emph{Pure} & \emph{Mixed} \\ \midrule
    \emph{A. Dominance} & \emph{Linearization} & 1 & 0 \\
    ($\exists k: \Delta_k(0)\cdot\Delta_k(1) > 0$) & Cycle acts as a Chain ($k \to k+1$) & {\tiny (Unique)} & {\tiny (Impossible)} \\ \midrule
    \emph{B. Global Switching} & \emph{Closed Feedback} &  &  \\
    ($\forall i: \Delta_i(0)\Delta_i(1) < 0$) &  &  &  \\
    \emph{B.1 Even Parity} & Bistability & \textbf{2} & \textbf{1} \\
    ($\kappa$ even)&  & {\tiny (Attractors)} & {\tiny (Unstable)} \\
    \emph{B.2 Odd Parity} & Frustration & \textbf{0} & \textbf{1} \\
     ($\kappa$ odd) &  & {\tiny (None)} & {\tiny (Unique)} \\ \bottomrule
  \end{tabular}
\end{table}

\subsection{Regime II: Degeneracy and Topological Singularities}

Relaxing the robustness assumption ($\exists k \in I :
\Delta_k(0)\cdot\Delta_k(1) = 0$) introduces \textit{Topological
  Singularities}. These nodes function not merely as unconditional
``circuit breakers'' (as in Structural Indifference), but also as
\textit{conditional bifurcation points} (in cases of Boundary
Indifference).

In the latter scenario, the topological decoupling is state-dependent:
the node becomes a source of indifference (a bifurcation) strictly
when induced by a specific predecessor strategy. Consequently, the
existence of such equilibria is not guaranteed locally but requires a
constructive verification of the global feedback consistency to ensure
the cycle creates the necessary triggering condition.

To efficiently construct the equilibrium set, we establish an
\emph{Operational Hierarchy}, distinguishing between unconditional
decoupling and hypothesis-driven verification.

\subsubsection{Priority 1: Unconditional Decouplers}

We define an \emph{Unconditional Decoupler} as a node $k$ where the
set of optimal strategies is invariant to the predecessor's input
$p_{k-1}$. These nodes do not physically sever the cycle, but they
functionally block the propagation of strategic constraints. By
ignoring the information contained in the predecessor's state for
decision-making purposes, they serve as axiomatic starting points to
resolve the cycle into linear sequences.

\begin{enumerate}
\item \emph{Dominance (Deterministic Source):} 
  \begin{itemize}
  \item \textit{Mechanism:} The marginal incentive maintains a
    constant non-zero sign ($\Delta_k(0) \cdot \Delta_k(1) >
    0$). Consequently, the optimal strategy $p_k^* \in \{0,1\}$ is
    strict and unique, regardless of the input $p_{k-1}$.
  \item \textit{Topological Effect:} The node acts as a
    \emph{Deterministic Source}. It effectively filters out upstream
    variability and initiates a deterministic propagation segment ($k
    \to k+1 \dots$) downstream.
  \end{itemize}

\item \emph{Structural Indifference (Entropic Source):} Arises when
  the marginal incentive vanishes globally ($\Delta_k \equiv 0$).
  \begin{itemize}
  \item \textit{Mechanism:} The player's payoff may still depend on
    the predecessor (shifting the utility baseline), but the
    \emph{marginal gain} from switching strategies is identically
    zero. Therefore, the optimal decision is decoupled from the input:
    the player is unconstrained and $p_k$ becomes a free parameter in
    $[0,1]$.
  \item \textit{Topological Effect:} The node functions as an
    \emph{Entropic Source}. It halts the propagation of upstream
    constraints and initiates a new solution manifold parameterized by
    $p_k$, which the algorithm must propagate downstream.
  \end{itemize}
\end{enumerate}

\subsubsection{Priority 2: Conditional Decouplers (Bifurcation Analysis)}
If no unconditional decoupler exists, the cycle structure relies on
nodes with \emph{Boundary Indifference}, creating a topological
bifurcation. By Lemma \ref{lem:dichotomy}, any equilibrium involving
node $k$ must fall into exactly one of two states: Indifference or
Determinism. We verify both branches exhaustively:

\begin{enumerate}
\item \emph{Branch A: The Indifference State ($\Delta_k = 0$).}
  \begin{itemize}
  \item \textit{Condition:} Requires the predecessor to adopt the
    trigger boundary strategy $p_{k-1}=p_{trigger}$ (the unique root
    of $\Delta_k$).
  \item \textit{Analysis:} Node $k$ becomes a free parameter $p_k \in
    [0,1]$. We propagate this parameter downstream, mirroring the
    logic of \emph{Structural Indifference}.
  \item \textit{Closure:} Since downstream nodes may also exhibit
    conditional indifference, the propagation generates a
    \emph{Solution Tree}. A valid equilibrium exists for any pair of
    parameter $p_k^*$ and trajectory branch where the cycle
    consistently returns the trigger strategy (i.e., the predecessor
    is either compelled or permitted to play $p_{trigger}$).
  \end{itemize}

\item \emph{Branch B: The Deterministic State ($\Delta_k \neq 0$).}
  \begin{itemize}
  \item \textit{Condition:} Applies whenever the predecessor operates
    in the \emph{non-trigger domain} ($p_{k-1} \neq p_{trigger}$).
  \item \textit{Analysis:} Node $k$ is compelled to adopt a unique
    pure strategy $s_{response} \in \{0,1\}$. This strategy is
    invariant for any input in this domain. We propagate
    $s_{response}$ downstream as a fixed source.
  \item \textit{Closure:} Valid equilibria exist if the return chain
    \emph{does not contradict} the non-trigger
    condition. Specifically, we verify that the cycle does not compel
    the predecessor to revert to the specific trigger strategy
    $p_{trigger}$. Any other return state (mixed or complementary
    pure) validates the equilibrium.
  \end{itemize}
\end{enumerate}

\textit{Operational Outcome:} Since Lemma \ref{lem:dichotomy}
precludes any third state, the union of equilibria found in Branch A
and Branch B constitutes the complete solution set.

\subsection{Algorithmic Resolution Procedure}
The resolution procedure is structured hierarchically. We first verify
the global incentive properties to determine if the system forms a
Robust Loop. If topological singularities exist, we proceed to
decompose the cycle via unconditional or conditional decouplers.

\begin{algorithm_logic}[Hierarchical Cycle Resolution]\label{alg:algorithm}\
  \begin{enumerate}
  \item \emph{State I: Robust Incentives Verification.}  Check if the
    system possesses a Robust Incentive Structure ($\forall i,
    \Delta_i(0)\cdot\Delta_i(1) \neq 0$).
    \begin{itemize}
    \item \emph{Step 1.1: Dominance Check.}  Does any player exhibit
      Dominance ($\Delta_k(0)\cdot\Delta_k(1) > 0$)?
      \begin{itemize}
      \item \textit{Result:} If Yes, there exists a \emph{Unique Pure
      Equilibrium}. The dominant node sets the state, triggering a
        deterministic cascade.
      \end{itemize}
    \item \emph{Step 1.2: Parity Check (No Dominance).}  If no
      dominance exists, the system is a closed loop of switching
      agents. Calculate $\kappa$ (number of Inverters).
      \begin{itemize}
      \item \emph{Even Parity:} Two pure equilibria (Bistability)
        and one fully mixed.
      \item \emph{Odd Parity:} Zero pure equilibria (Frustration) and
        one unique fully mixed.
      \end{itemize}
    \end{itemize}
  \item \emph{State II: Non-Robust Systems (Topological
  Decomposition).}  If the system is not Robust ($\exists i,
    \Delta_i(0)\cdot\Delta_i(1) = 0$), the cycle contains
    singularities.
    
    \begin{itemize}
    \item \emph{Step 2.1: Unconditional Decouplers.}  Scan for nodes
      that sever the feedback loop regardless of the input.
      \begin{itemize}
      \item \textit{Targets:} Identify nodes with \emph{Dominance}
        (even if local) OR \emph{Structural Indifference} ($\Delta_i
        \equiv 0$).
      \item \textit{Action:} These nodes act as Sources. Cut the cycle and construct the equilibrium by concatenating the linear segments generated by forward propagation.
      \end{itemize}
      
    \item \emph{Step 2.2: Conditional Bifurcation (No Decouplers).}
      If no unconditional decouplers exist, the system relies on
      \emph{Boundary Indifference}.
      \begin{itemize}
      \item \textit{Action:} Select a singular node $k$ exhibiting
        singularity at a boundary and propagate two branches:
        \begin{enumerate}
        \item \emph{Hypothesis A:} Predecessor triggers indifference.
        \item \emph{Hypothesis B:} Predecessor does not trigger
          (Deterministic response).
        \end{enumerate}
      \item \textit{Closure:} Verify a coherent return for each
        branch. The solution set is the union of valid profiles.
      \end{itemize}
    \end{itemize}
  \end{enumerate}
\end{algorithm_logic}

\section{Illustrative Applications: Inventory Synchronization in Circular Supply Chains}
\label{sec:applications}

To validate our theoretical framework within a core Operations
Research context, we analyze a \emph{Strategic Inventory Policy
Game}. While classical inventory models often optimize continuous
quantities, strategic network design requires discrete choices between
contrasting operational regimes (e.g., Push vs. Pull, Lean vs. Agile)
\cite{CachonNetessine2004}.

We model a decentralized Closed-Loop Supply Chain (CLSC) of $N=4$
agents (e.g., Manufacturer $\to$ Retailer $\to$ Collector $\to$
Recycler $\to$ Manufacturer). In this sequential game, the inventory
policy of an upstream partner ($i-1$) strictly conditions the supply
risk and holding costs of the downstream agent ($i$).

\subsection{Problem Formulation}
Each agent $i$ selects a steady-state inventory strategy $s_i \in \{0,
1\}$:
\begin{itemize}
\item \emph{Strategy 0 (Lean Policy / JIT):} The agent minimizes
  inventory buffers to reduce \emph{Holding Costs} and working
  capital, prioritizing cost efficiency over availability.
\item \emph{Strategy 1 (Bulk Policy / Safety Stock):} The agent
  maintains high inventory levels to maximize \emph{Service Level} and
  hedge against supply volatility, prioritizing availability over
  capital efficiency.
\end{itemize}

\subsection{Agent Taxonomy}
We define five archetypal agent profiles. Each maps a specific
operational constraint to one of the topological classes identified in
our theoretical framework (Definition \ref{def:taxonomy} and Lemma
\ref{lem:incentive_regimes}).

\begin{enumerate}
\item \emph{The Market Stabilizer ($A_{inv}$):} A distributor
  practicing counter-cyclical management (Strategic Substitutes).
  \begin{itemize}
  \item \emph{Logic:} If Upstream is Lean (0), scarcity drives margins
    up $\to$ Adopt \emph{Bulk} (1). If Upstream is Bulk (1),
    saturation risks obsolescence $\to$ Adopt \emph{Lean} (0).
  \item \emph{Matrix:} $H_{inv} = \begin{pmatrix} 2 & 5 \\ 4 &
    1 \end{pmatrix}$. ($\Delta(0)=+3, \Delta(1)=-3$).
  \item \emph{Class: Inverter} (Switching Agent).
  \end{itemize}
  
\item \emph{The Integrated Partner ($A_{conf}$):} A manufacturer in a
  synchronized JIT ecosystem (Strategic Complements).
  \begin{itemize}
  \item \emph{Logic:} If Upstream is Lean (0), the agent lacks raw
    materials for high-volume production $\to$ Forced \emph{Lean} (0)
    to avoid line starvation costs. If Upstream is Bulk (1), supply
    abundance allows economies of scale $\to$ Adopt \emph{Bulk} (1).
  \item \emph{Matrix:} $H_{conf} = \begin{pmatrix} 4 & 1 \\ 2 & 5 \end{pmatrix}$. ($\Delta(0)=-3, \Delta(1)=+3$).
  \item \emph{Class: Conformist} (Synchronization Agent).
  \end{itemize}

\item \emph{The Constrained Warehouse ($A_{dom}$):} An agent with
  severe physical storage limitations (e.g., an urban kiosk).
  \begin{itemize}
  \item \emph{Logic:} Adopting \emph{Bulk} (1) necessitates expensive
    external leasing. Under Lean supply, this incurs wasteful rental
    fees ($\Delta(0) < 0$); under Bulk supply, it triggers severe
    capital lock-up and obsolescence risks ($\Delta(1) \ll
    0$). Strategy \emph{Lean} (0) is strictly dominant.
  \item \emph{Matrix:} $H_{dom} = \begin{pmatrix} 3 & 1 \\ 3 & -5 \end{pmatrix}$. ($\Delta < 0$ always).
  \item \emph{Class: Strict Dominance} (Decoupler).
  \end{itemize}

\item \emph{The Logistics Hub ($A_{ind}$):} A cross-docking
  facility. Revenue derives strictly from flow volume (transaction
  fees), not from holding margins.
  \begin{itemize}
  \item \emph{Logic:} The agent's utility is highly sensitive to the
    upstream volume (Input 0 vs 1), but the operational cost of
    declaring \emph{Lean} or \emph{Bulk} is identical ($C \approx
    0$). Thus, the agent is operationally indifferent to its own
    label, breaking the strategic dependency.
  \item \emph{Matrix:} $H_{ind} = \begin{pmatrix} 1 & 1 \\ 3 &
    3 \end{pmatrix}$. ($\Delta(0)=0, \Delta(1)=0$).
  \item \emph{Class: Structural Indifference} (Entropic Source, Decoupler).
  \end{itemize}

\item \emph{The Flow-Constrained Retailer ($A_{bound}$):} An agent
  relying on Just-in-Time fulfillment to minimize holding costs.
  \begin{itemize}
  \item \emph{Logic:}
    \begin{itemize}
    \item \emph{Input 0 (Scarcity):} If Upstream is Lean, inventory
      accumulation is physically impossible due to lack of
      flow. Strategies 0 and 1 yield identical low outcomes
      ($\Delta(0)=0$).
    \item \emph{Input 1 (Abundance):} If Upstream is Bulk, supply is
      guaranteed. The agent strictly prefers \emph{Lean} (0) to
      minimize costs, effectively using the supplier's inventory as a
      buffer. Adopting \emph{Bulk} (1) would incur redundant storage
      costs for no additional sales gain ($\Delta(1) < 0$).
    \end{itemize}
  \item \emph{Matrix:} $H_{bound} = \begin{pmatrix} 2 & 2 \\ 5 & 1 \end{pmatrix}$.
  \item \emph{Class: Boundary Indifference} (Conditional Decoupler).
  \end{itemize}  
\end{enumerate}

\subsection{Scenario 1: Robust Dominance (The Reverse Logistics Loop)}

We model a heterogeneous recycling chain where physical constraints
dictate the system's state. The cycle consists of:
\begin{itemize}
\item \emph{Node 2 ($A_{dom}$):} An urban collection point with no
  storage capacity (Bottleneck).
\item \emph{Node 3 ($A_{inv}$):} A regional consolidator managing
  supply buffers.
\item \emph{Node 4 ($A_{conf}$):} A processing plant seeking economies
  of scale.
\item \emph{Node 1 ($A_{inv}$):} A distributor sensitive to market
  saturation.
\end{itemize}
\begin{equation*}
  \text{Topology:} \quad 1(A_{inv}) \to 2(A_{dom}) \to 3(A_{inv}) \to 4(A_{conf}) \to 1
\end{equation*}

The system exhibits a \emph{Robust Incentive Structure} with a
Dominant node ($A_{dom}$). By \emph{Theorem
\ref{thm:dominance_uniqueness}}, this guarantees a \emph{unique pure
strategy equilibrium}. Unlike iterative algorithms that might cycle
indefinitely, our topological framework resolves the equilibrium in
linear time by propagating the deterministic constraint from the
source (Proposition \ref{prop:propagation}):

\begin{equation*}
  \underbrace{s_2^*=0}_{\substack{\text{Bottleneck} \\ \text{Constraint}}} 
  \xrightarrow[\text{Scarcity}]{A_{inv}} 
  \underbrace{s_3^*=1}_{\substack{\text{Buffer} \\ \text{Accumulation}}} 
  \xrightarrow[\text{Abundance}]{A_{conf}} 
  \underbrace{s_4^*=1}_{\substack{\text{Mass} \\ \text{Production}}} 
  \xrightarrow[\text{Saturation}]{A_{inv}} 
  \underbrace{s_1^*=0}_{\substack{\text{Risk} \\ \text{Mitigation}}}
\end{equation*}

The system linearizes into the unique profile $s^* = (0, 0, 1, 1)$.

\textit{Operational Insight:} The capacity constraint at the
collection point (Node 2) acts as a "Dictator," structurally
stabilizing the entire supply chain. Policy interventions (e.g.,
subsidies) targeting the factory (Node 4) would be ineffective, as its
high-production state is a deterministic consequence of the
bottleneck's dominance.

\subsection{Scenario 2: Robust Loop (The Bullwhip Stabilizer)}

We analyze a homogeneous supply ring of four \emph{Market Stabilizers
($A_{inv}$)}. Every agent acts as a strategic substitute, creating a
closed feedback loop with no dominance.
\begin{equation*}
  \text{Topology:} \quad A_{inv} \to A_{inv} \to A_{inv} \to A_{inv} \to 1
\end{equation*}

The system possesses a Robust Incentive Structure with an even number
of inverters ($\kappa=4$). Applying our theoretical framework yields
the complete set of equilibria:
\begin{enumerate}
\item \emph{Pure Strategies:} By Theorem \ref{thm:parity_condition}
  (Even Parity), the cycle supports two distinct stable waves:
  \begin{equation*}
    s_A^* = (0, 1, 0, 1) \quad \text{and} \quad s_B^* = (1, 0, 1, 0)
  \end{equation*}
\item \emph{Mixed Strategies:} By Theorem \ref{thm:mixed_equilibrium},
  a unique fully mixed equilibrium exists where agents randomize to
  induce downstream indifference. Using Eq. \eqref{eq:mixed_algo} with
  inputs from $H_{inv}$:
  \begin{equation*}
    p^* = (0.5, 0.5, 0.5, 0.5)
  \end{equation*}
\end{enumerate}

To determine the operational viability of these solutions, we compute
the payoff vectors $\pi = (\pi_1, \pi_2, \pi_3, \pi_4)$ for each
profile:
\begin{itemize}
\item For the pure waves $s_A^*$ and $s_B^*$, agents alternate between
  maximizing margin ($H(0,1)=5$) and minimizing risk ($H(1,0)=4$).
  \begin{equation*}
    \pi(s_A^*) = (4, 5, 4, 5) \quad \text{and} \quad \pi(s_B^*) = (5, 4, 5, 4)
  \end{equation*}
\item For the mixed profile $p^*$, the expected payoff is strictly
  limited by the indifference condition:
  \begin{equation*}
    \pi(p^*) = (3, 3, 3, 3)
  \end{equation*}
\end{itemize}

\textit{Operational Insight:} Comparing the payoff tuples reveals that
the mixed equilibrium is \emph{strictly Pareto-dominated} by either
pure configuration ($\pi(p^*) \ll \pi(s_{pure}^*)$). Consequently,
although the mixed profile constitutes a theoretically valid
equilibrium, it lacks \emph{operational viability}. Under standard
equilibrium selection criteria based on payoff dominance, rational
profit-maximizing agents will inevitably coordinate on a deterministic
wave, relegating the mixed equilibrium to an unstable transient state.

\subsection{Scenario 3: Structural Indifference (The Control Tree)}

We examine a chain initiated by a \emph{Logistics Hub ($A_{ind}$)}
followed by three \emph{Market Stabilizers ($A_{inv}$)}. This topology
represents a supply chain driven by a volume-agnostic 3PL provider.
\begin{equation*}
  \text{Topology:} \quad 1(A_{ind}) \to 2(A_{inv}) \to 3(A_{inv}) \to 4(A_{inv}) \to 1
\end{equation*}
Agent 1 ($A_{ind}$) exhibits Structural Indifference ($\Delta_1 \equiv
0$); consequently, its strategy $p_1 \in [0,1]$ serves as an
unconstrained exogenous parameter. Given that the subsequent player
(Agent 2) operates under Robust Incentives without dominance
($A_{inv}$), $p_1$ functions as a control mechanism capable of
inducing three distinct topological states downstream:
\begin{enumerate}
\item \emph{Deterministic Bulk:} If $p_1 < 0.5$, Agent 2 is compelled
  to play \emph{Bulk} ($s_2=1$), locking the downstream sequence into
  a fixed pure trajectory.
\item \emph{Deterministic Lean:} If $p_1 > 0.5$, Agent 2 is compelled
  to play \emph{Lean} ($s_2=0$), locking the sequence into the
  complementary pure trajectory.
\item \emph{Induced Indifference:} If $p_1 = 0.5$, Agent 2 becomes
  indifferent. This grants Agent 2 strategic degrees of freedom ($p_2
  \in [0,1]$), generating a \emph{Solution Tree}.
\end{enumerate}

Since the logic for subsequent agents applies recursively, the
complete equilibrium manifold is exhaustively mapped in Table
\ref{tab:cascade_map}.

\begin{table}[htbp]
  \centering
  \caption{System Response Map (Equilibrium Manifold) for Scenario 3}
  \label{tab:cascade_map}
  \begin{tabular}{cccc}
    \toprule
    Agent 1 ($p_1$) & Agent 2 ($p_2$) & Agent 3 ($p_3$) & Agent 4 ($p_4$) \\
    \midrule
    $[0, 0.5)$ & 1 & 0 & 1 \\
    $(0.5, 1]$ & 0 & 1 & 0 \\
    \midrule
    $0.5$ & $[0, 0.5)$ & 1 & 0 \\
    $0.5$ & $(0.5, 1]$ & 0 & 1 \\
    \midrule
    $0.5$ & $0.5$ & $[0, 0.5)$ & 1 \\
    $0.5$ & $0.5$ & $(0.5, 1]$ & 0 \\
    \midrule
    $0.5$ & $0.5$ & $0.5$ & $[0,1]$ \\
    \bottomrule
  \end{tabular}
\end{table}

\textit{Operational Insight (Structural Agility):} This hierarchy
offers a mechanism for \emph{Agile Network Control}. By initializing
the source at the critical value ($p_1=0.5$), the central planner
places the system in a state of maximum entropy ("Full Cascade"),
preserving \emph{degrees of freedom} for downstream agents. This
grants the chain the flexibility to absorb unmodeled externalities: if
Node 3 faces a sudden demand shock, it can unilaterally "collapse" the
wave into a deterministic response without violating the global
equilibrium conditions established by the predecessor. Thus, the
indifference cascade acts as a buffer of \emph{deferred commitment},
allowing the network to remain fluid until a specific constraint
requires rigidity.

\subsection{Scenario 4: Boundary Indifference (Algorithmic Resolution)}

We address a non-robust configuration containing a
\emph{Flow-Constrained Retailer ($A_{bound}$)} at Node 2. This agent
introduces a topological singularity: its behavior is deterministic
under abundance but indifferent under scarcity.

\begin{equation*}
  \text{Topology:} \quad 1(A_{inv}) \to 2(A_{bound}) \to 3(A_{inv}) \to 4(A_{inv}) \to 1
\end{equation*}

Since Node 2 exhibits a singularity at Input 0 ($\Delta_2(0)=0$), we
apply \emph{Algorithm \ref{alg:algorithm}} to resolve the bifurcation
by hypothesizing the predecessor's state (Node 1) and verifying the
coherence of the resulting propagation.

\begin{enumerate}
\item \emph{Branch A: The Indifference Hypothesis.} Assume the
  predecessor plays the trigger strategy $s_1=0$.
  \begin{itemize}
  \item \emph{Consequence:} Node 2 becomes a free parameter $p_2 \in
    [0,1]$.
  \item \emph{Verification:} This freedom generates a propagation
    tree. We validate only those branches where the cycle consistently
    returns the trigger value $s_1=0$ (or induces indifference
    allowing it). Table \ref{tab:scenario4_check} maps the valid
    domains.
  \end{itemize}
  
  \begin{table}[htbp]
    \centering
    \caption{Consistency Verification for Hypothesis $s_1=0$}
    \label{tab:scenario4_check}
    \begin{tabular}{c c c c c}
      \toprule
      \textbf{Node 2 ($p_2$)} & \textbf{Node 3 ($p_3$)} & \textbf{Node 4 ($p_4$)} & \textbf{Node 1 (Return)} & \textbf{Status} \\
      \midrule
      $[0, 0.5)$ & 1 & 0 & 1 & \emph{Invalid} \\
      $(0.5, 1]$ & 0 & 1 & \textbf{0} & \textbf{Valid} \\
      \midrule
      $0.5$ & $[0, 0.5)$ & 1 & \textbf{0} & \textbf{Valid} \\
      $0.5$ & $(0.5, 1]$ & 0 & 1 & \emph{Invalid} \\
      \midrule
      $0.5$ & $0.5$ & $[0, 0.5)$ & 1 & \emph{Invalid} \\
      $0.5$ & $0.5$ & $(0.5, 1]$ & \textbf{0} & \textbf{Valid} \\
      \midrule
      $0.5$ & $0.5$ & $0.5$ & \textbf{0} (via Indiff.) & \textbf{Valid} \\
      \bottomrule
    \end{tabular}
  \end{table}
  
\item \emph{Branch B: The Deterministic Hypothesis.} Assume the
  predecessor operates in the non-trigger domain ($s_1 \in (0, 1]$).
    \begin{itemize}
    \item \emph{Consequence:} Since the input is strictly non-zero,
      Node 2 faces a negative incentive ($\Delta_2 < 0$) and is
      compelled to adopt \emph{Lean} ($s_2=0$).
    \item \emph{Verification:} This deterministic response propagates
      through the chain of inverters:
      \begin{equation*}
        s_1 \in (0, 1] \xrightarrow{} s_2=0 \xrightarrow{A_{inv}} s_3=1 \xrightarrow{A_{inv}} s_4=0 \xrightarrow{A_{inv}} \mathbf{s_1=1}
      \end{equation*}
    \item \emph{Result:} The propagation process selects $s_1=1$ as
      the unique fixed point. The only valid profile in this branch is
      $s_B^* = (1, 0, 1, 0)$.
    \end{itemize}
\end{enumerate}

The system exhibits \emph{Regime Asymmetry}. Under scarcity ($s_1=0$),
the solution set expands into a \emph{Hierarchical Solution Tree}
allowing for pure, hybrid, or mixed configurations (as detailed in
Table \ref{tab:scenario4_check}). Under abundance ($s_1=1$), the
solution collapses into a single deterministic point.

\textit{Operational Insight (Strategic Plasticity vs. Structural
  Locking):} This result highlights a fundamental property of supply
chains with boundary singularities. The "Lean" regime ($s_1=0$) is
characterized by \emph{Strategic Plasticity}: the system can absorb
local variations (entropy) from Agent 2 without breaking the
equilibrium loop. Conversely, the "Bulk" regime ($s_1=1$) is
characterized by \emph{Structural Locking}: any deviation is instantly
corrected by the feedback loop, forcing the system back to the rigid
pure trajectory. Thus, the Lean state is operationally richer but
potentially more complex to manage, while the Bulk state is simpler
but rigid.

\section{Conclusions}
\label{sec:conclusions}

This research provides a closed-form characterization of equilibrium
structures in $N$-player non-cooperative games governed by
unidirectional local dependencies. Addressing the computational
intractability (PPAD-completeness) inherent to general network games,
we exploit the topological modularity of the dyadic cycle to derive a
deterministic resolution framework with linear time complexity $O(N)$.

Crucially, this reduction to binary strategies and sequential
constraints represents a \emph{methodological distillation}, not a
trivialization. As evidenced by the Supply Chain applications, complex
operational decisions frequently crystallize into discrete regimes
(e.g., Lean vs. Bulk) driven by strict upstream dependencies. By
abstracting from the noise of generalized connectivity, our
topological approach isolates the fundamental causal mechanics that
remain opaque in black-box numerical solvers.

Our analysis classifies the equilibrium landscape into two distinct
topological regimes:
\begin{enumerate}
\item \emph{Regime I: Robust Incentives (The Generic Operational
State).}  In systems where marginal incentives are strictly non-zero
  at the boundaries --representing the standard economic condition of
  strict preferences-- the global equilibrium structure is determined
  by the presence of dominance:
  \begin{itemize}
  \item \emph{Dominance (Linearization):} We proved that local
    physical constraints (e.g., storage bottlenecks) act as "Circuit
    Breakers" that functionally sever the feedback loop. This
    transforms the cyclic problem into a deterministic forward
    propagation, guaranteeing a unique pure strategy equilibrium
    (Theorem \ref{thm:dominance_uniqueness}).
  \item \emph{Feedback Loops (Parity):} In the absence of dominance,
    stability is strictly governed by the \emph{Parity of Inverters}
    (Theorem \ref{thm:parity_condition}), determining the existence of
    exactly two or zero pure strategy equilibria. Interestingly, a
    unique fully mixed equilibrium is structurally guaranteed (Theorem
    \ref{thm:mixed_equilibrium}). Notably, our case study reveals that
    this mixed profile is strictly Pareto-dominated by the pure
    strategy equilibria.
  \end{itemize}

\item \emph{Regime II: Topological Singularities (Degenerate Cases).}
  When incentives vanish locally, the system admits degrees of freedom
  that require specific algorithmic resolution:
  \begin{itemize}
  \item \emph{Structural Indifference (Inverse Design):} Agents with
    invariant incentives (e.g., Logistics Hubs) serve as control
    parameters. A central planner can tune these source nodes to
    induce controlled flexibility (entropy) throughout the chain,
    enabling "Agile" network configurations.

  \item \emph{Boundary Indifference (Bifurcation):} Agents with
    conditional incentives (e.g., Contracted Retailers) require a
    \emph{Hypothesis-Verification} analysis. Our case study highlights
    a "Regime Asymmetry" where specific boundary conditions (e.g.,
    scarcity) enable strategic plasticity (multiple equilibria),
    whereas complementary conditions (e.g., abundance) enforce
    structural locking (unique deterministic equilibrium).
  \end{itemize}
\end{enumerate}

From a managerial perspective, this framework offers rigorous tools
for the structural design of circular supply chains. It enables
practitioners to pinpoint the local parameters required to stabilize
oscillating networks or to engineer specific load-balancing profiles
via \emph{Inverse Design}. Future research should extend this
topological approach to bidirectional dependencies ($i$ depends on
both $i-1$ and $i+1$) and non-binary strategy spaces, investigating
the boundaries of deterministic propagation beyond the dyadic binary
case.

\section*{Disclosure of interest}

The authors report there are no competing interests to declare.

\section*{Declaration of Generative AI and AI-assisted technologies in the writing process}

During the preparation of this work, the authors used Gemini in order
to improve the readability, language quality, and structural flow of
the manuscript. After using this tool/service, the authors reviewed
and edited the content as needed and take full responsibility for the
content of the publication.

\section*{CRediT authorship contribution statement}

\emph{José-María Sánchez-Sáez:} Conceptualization, Methodology, Formal
analysis, Writing -- original draft.

\emph{Nana Odishelidze:} Methodology, Formal analysis, Validation,
Writing -- review \& editing.

\emph{Francisco Criado-Aldeanueva:} Investigation, Validation,
Visualization, Writing -- review \& editing, Project administration.

%\bibliographystyle{elsarticle-num}
%\bibliography{criad006}

\end{document}